\documentclass[12pt]{amsart}
\usepackage{amsmath,amsfonts,amssymb} \usepackage{color}
\usepackage[all,cmtip]{xy}
\usepackage{graphicx}

\newcommand{\corr}[1]{\langle {#1} \rangle}

\newcommand{\bt}{{\bf t}}

 \newcommand{\pd}{\partial}
\newcommand{\Mbar}{\overline{\mathcal M}}

 \DeclareMathOperator{\Deg}{Deg}   \DeclareMathOperator{\Tr}{Tr}

\newcommand{\be}{\begin{equation}}
\newcommand{\ee}{\end{equation}}
\newcommand{\bea}{\begin{eqnarray}}
\newcommand{\eea}{\end{eqnarray}}
\newcommand{\ben}{\begin{eqnarray*}}
\newcommand{\een}{\end{eqnarray*}}

\newcommand{\half}{\frac{1}{2}}

\newtheorem{cor}{Corollary}[section]

 \newtheorem{prop}[cor]{Proposition}
 \newtheorem{thm}[cor]{Theorem}
\theoremstyle{remark}

 \newtheorem{rmk}[cor]{Remark}

\definecolor{A}{rgb}{.75,1,.75}

\definecolor{yellow}{rgb}{1,1,0}
\definecolor{orange}{rgb}{1,.7,0}
\definecolor{red}{rgb}{1,0,0}
\definecolor{white}{rgb}{1,1,1}
\definecolor{green}{rgb}{0,1,0}

 \makeindex
\begin{document}
\title
{On Phase Transition of Two-Dimensional Topological Gravity}
\author{Jian Zhou}


\begin{abstract}
We show that one can use some renormalized coupling constants to compute the free energy and correlation functions
at all critical points of the two-dimensional topological gravity in a uniform way.
In particular,
one can derive the critical exponents of the free energy and correlation functions at all critical points in a uniform way.
Some concrete results for the case of $(3,2)$-model (pure gravity) and the $(5,2)$-model (Yang-Lee edge singularity coupled with gravity) 
are also presented.
\end{abstract}

\maketitle

\section{Introduction}

There are several different approaches to two dimensional quantum gravity theory in the physics literature.
They are based on different kinds of geometries of surfaces,
and they lead to different fields of mathematics.

The first approach is the differential geometric approach of Liouville field theory  (see e.g. \cite{David, Distler-Kawai}, for a review see \cite{Nakayama})
in which the action involves the Liouville field, the conformal matter fields, and the ghost fields.
A milestone in this approach is the fractional structures discovered in \cite{KPZ} and generalized in \cite{David, Distler-Kawai}.
These are the salient features that have to be reproduced by other approaches. 
The conformal fields are usually taken from the $(p,q)$-minimal model,
where $p$ and $q$ are coprime.
They have central charge
\ben
&& c = 1 - \frac{6(p-q)^2}{pq}, 
\een
which gives the dimension of a fictitious space used to define the Liouville action.  
It can be a rational number or even a negative number.
The  primary fields $\Phi_{r,s}$ are indexed by $(r,s)$  ($1 \leq r < q$, $1 \leq s < p$),
and the conformal dimensions of $\Phi_{r,s}$ are
\ben
h_{r,s} = \frac{(pr-qs)^2-(p-q)^2}{4pq}.
\een 
The Liouville fields are concerned with the Riemannian metrics on Riemann surfaces,
so the Liouville field theory depends on the area $A$ of the surfaces.
The free energy in genus $g$ takes the form
\be 
F_g(A) \sim K \cdot A^{\gamma_g-3},
\ee
where $\gamma_g$ is the {\em string susceptibility} given by the KPZ formula 
\be 
\gamma_0 = \frac{1}{12} (c-1 - \sqrt{(1-c)(25-c)})
\ee
in genus zero \cite{KPZ}, 
and by 
\be 
\gamma_g  = (1-g)(\gamma_0-2)+2 = \frac{1-g}{12} (c-25- \sqrt{(1-c)(25-c)} +2   
\ee
in higher genera \cite{David, Distler-Kawai}.
After coupling with gravity, 
a primary field $\phi_0$ with conformal weight $\Delta_0$ is gravitationally dressed to become an operator $\phi=\phi_0e^{\alpha \varphi}$,  
where 
\be 
\alpha = \frac{1}{\sqrt{12}} (\sqrt{25-c} - \sqrt{1-c+24 \Delta_0}),
\ee
then the dressed primary operator has scaling exponent:
\be 
\Delta = \frac{\sqrt{1-c+24\Delta_0}-\sqrt{1-c}}{\sqrt{25-c} -\sqrt{1-c}}.
\ee
This is the DDK/KPZ formula. 
This formula means if $\phi^1, \dots, \phi^n$ are $n$ dressed operators with scaling exponents $\Delta_1, \dots, \Delta_n$ respectively,
then the $n$-point correlation function has the following behavior:
\be 
\corr{\phi^1 \cdots \phi^n}_g \sim A^{\gamma_g -3 + \sum_{i=1}^n (1- \Delta_i)}.
\ee

So in  the Liouville field approach to two-dimensional gravity,
one has different models, indexed by coprime couples of integers  $(p,q)$,
each model has some fraction structure and  one is concerned with the critical exponents of the free energy and the scaling operators.
To unify them, it is natural to regard them as different phases of a single theory 
and study the phase transitions in this theory,
similar to the Wilson paradigm of the study of phase transitions and critical phenomena in statistical physics in
which a big phase space of infinitely many coupling constants is considered.
Starting with a single model and considering its perturbations by adding to its action all possible scaling operators,
one is supposed to be able to reach all other models as (multi)critical points.
This was indeed the point of view taken in \cite{Distler}:
Starting with the model $(1,2)$,  all models $(2k-1,2)$ can be reached.

The second approach to two-dimensional quantum gravity is the discrete geometric approach of dynamical tessellations of Riemann surfaces.
For reviews, see e.g. \cite{DGZ, GM}.
In this approach one starts with a matrix model which is integrations of suitable weight function
over some spaces of matrices, 
and interprets the partition function as a suitable way of counting tessellations 
of surfaces.
One can study the partition functions and correlations functions in matrix models 
using the theory of orthogonal polynomials.
This fact plays a crucial role in the following two aspects in the development of the matrix model approaches.
First, the three-term recursion relation for orthogonal polynomials  leads naturally  to the theory of integrable hierarchies, 
secondly, 
the asymptotic analysis of the orthogonal polynomials lead to the double scaling limits of matrix models
in which conformal matters emerge in the matrix model approach.
Kazakov \cite{Kazakov} observed that by fine tuning the potential function in
the matrix model,
one may get a multicritical model with string susceptibility 
\be 
\gamma_0 = - \frac{1}{k},
\ee 
he then suggested that this corresponds to the unitary minimal model $(k+1,k)$ coupled with two-dimensional gravity.
As discovered in later work,
his multicritical models actually correspond to the minimal $(2k-1,2)$-models coupled with gravity.
This will be referred to as the {\em $k$-th multicritial point}.
See e.g. \cite{Gross-Migdal} and the references therein. 
 One of the main results in \cite{Gross-Migdal} is as follows:
at the $k$-th critical point, 
the specific heat $u = F''(t)$ satisfies the following differential equation (see also \cite{Douglas-Shenker, Brezin-Kazakov}):
\be \label{eqn:k-th}
\begin{split}
t & = \frac{k!}{(2k-1)!!} \hat{K}[u(t), \pd_t]^k \cdot 1, \\
 \hat{K}[u(t), \pd_t] &:=- \half \pd_t^2 + u(t) + \pd_t^{-1}u(t) \pd_t,
\end{split}
\ee
where the operator $\hat{K}$ was introduced by Gelfand and Dikiï in their study of higher
order KdV equations \cite{Gelfand-Dikii} . 
Let us recall some key steps in their derivations for matrix models with even potentials.
First in the large $N$ limit, 
\be 
Z_N(g-s) \sim \exp \biggl( g_s^2 \int_0^t dx (t-x) \ln R(x)\biggr),
\ee
where $t = Ng_s$ is the 't Hooft coupling constant,
and $R(x)$ is determined from the string equation:
\be  \label{eqn:string}
-x + W(R(x)) = 0,
\ee
where $W(R)$ is determined from the even potential $U(\phi)$ as follows:
\begin{align}
U(\phi) & = \sum_k U_{2k} \phi^{2k}, & 
W(R) & = \sum \frac{(2k)!}{k!(k-1)!} U_{2k} R^k.
\end{align}
The $k$-th multicritical point occurs when $1 - W(R)$ and
$k - 1$ of its derivatives vanish at, say, $R = 1$.
In other words,
\begin{align}
W_k(R) & = 1- (1-R)^k, & R & = 1- (1-x)^{1/k}, \\ 
U_k(\phi) & = \sum_{j=1}^k (-1)^{j-1} \frac{k!(j-1)!}{k-j)!(2j)!} \phi^{2j}.
\end{align}
See also an earlier derivation in \cite{Kazakov} using loop equations in matrix model theory.
In genus zero one then gets
\be 
u_0(t) = F''_0(t) = t^{1/k},
\ee
where $t = 1-x$ up to a suitable rescaling constant,
and so 
\be 
F_0(t) = \frac{k^2}{(k+1)(2k+1)} t^{2+1/k},
\ee
ignoring the integration constants.
In order to get correlation functions in genus zero,
one needs to consider the deformed  model by introducing scaling operators $O_l$.
The operator $O_l$ introduced in \cite{Gross-Migdal} is $\Tr (2-\phi)^{l+1/2}$ up to some normalization constant in matrix model.
Then the following equation (called the string equation) is derived after taking large N limit in \cite{Gross-Migdal}:
\be \label{eqn:String-g=0}
t= u^k - \sum_i \mu_i u^i.
\ee
After solving for $u$ as a function of $t$ and the $\mu_i$'s and integrating twice,
one can obtain the free energy in genus zero, again ignoring all the integration constants.
The following explicit formula is obtained in \cite{Gross-Migdal}:
\be 
\begin{split}
\corr{O_{l_1} \cdots O_{l_p}}_0 
& = - \frac{\pd}{\pd \mu_{l_1}} \cdots \frac{\pd}{\pd \mu_{l_p}} F(t; \mu_0, \mu_1, \dots) \biggl|_{\mu_j =0} \\
& = - \frac{1}{k} (\frac{\pd}{\pd t})^{p-3} t^{(\sum_i l_i +1-k)/k} \\
& \sim t^{[2+1/k]+\sum_i (l_i/k-1)},
\end{split}
\ee 
These results are compatible with the DDK/KPZ formula,
so the fractional structures naturally emerge in the formalism of the string equation.
In the computations of  the free energy and the correlation functions in higher genera,
the operator $\hat{K}$ of Gelfand and Dikii  emerge naturally.
In the double scaling limit, $\Tr (2-\phi)^{l+1/2}$ corresponds to $L^{l+1/2}$, 
where $L = g_s^2\pd_x^2 + u$ is the Lax operator.
Using the results of Gelfand and Dikii,
\eqref{eqn:k-th} was then obtained. 
Since one can obtain all $(2k-1,2)$-models by fine tuning the coefficients in the even potential $U(\phi)$ after double scaling limit,
as proposed in \cite{Banks-Douglas-Seiberg-Shenker},
there arise the general massive model interpolating all the $(2k-1, 2)$-models,
so that the string equation \eqref{eqn:String}  takes the form
\be \label{eqn:String2}
x = \sum_{k=0}^\infty (k+\half) T_k R_k[u],
\ee 
where $u$ is the specific heat which satisfies 
\be 
\frac{\pd}{\pd T_k} u = \frac{\pd}{\pd x} R_k[u].
\ee
I.e., the partition function $Z$ is a tau-function of the KdV hierarchy specified by the string equation \eqref{eqn:String2}. 
The macroscopic loop operator $w(l)$ used in \cite{Gross-Migdal, Banks-Douglas-Seiberg-Shenker} has an expansion
\be 
w(l) = \sum_{n=0}^\infty \frac{l^{n+1/2}}{\Gamma(n+\frac{3}{2})} \sigma_n.
\ee
It satisfies the loop equation \cite{DVV}:
\be 
\pd_x^2 \frac{\pd}{\pd l} \corr{w(l)} 
= ( \half g_s^2 \pd_x^4 + 2u\pd_x^2 +(\pd_x u) \pd_x)\corr{w(x)} + \frac{\pd_x u}{\sqrt{\pi l}},
\ee
which can be reformulated as the Virasoro constraints satisfied by the partition function. 
According to  \cite{Douglas} the $(q - 1)$-matrix model is related to the p-th generalized
KdV hierarchy in which the p-th multi-critical point  corresponds to the $(p, q)$ minimal CFT coupled to gravity.
The partition functions of these models satisfy the $W$-constraints \cite{DVV, Fukuma-Kawai-Nakayama}.
Since all the $(q,p)$-models can be obtained by taking the double scaling limit by fine tuning the coupling constants in 
the potential of the $(q-1)$-matrix model,
one can again regard all $(p,q)$-models as different phases of a unified model.

More recently, two further developments have grown out of the matrix model approach.
The first is the Eynard-Orantin topological recursion \cite{Eynard-Orantin},
and the second is the determinantal formula for differential systems \cite{Bergere-Eynard}. 
They have both been applied the $(p,q)$-minimal models coupled with gravity and shown to yield the same results \cite{Bergere-Eynard2, Bergere-Borot-Eynard}.
 
The third approach is Witten's two dimensional topological gravity \cite{Witten-TG}.
He related the partition function in the $c=-2$ case to intersection numbers of 
Deligne-Mumford moduli spaces $\Mbar_{g,n}$ of stable algebraic curves  \cite{Witten-phase}.
Based on such connection to  algebraic geometry,
he derived some recursion relations in genus zero and genus one called the topological recursion relations.
In genus zero he used the topological recursion relations to recover the results of Gross-Migdal \cite{Gross-Migdal}.
Based on the conjectural equivalence of topological gravity approach and the matrix model approach,
Witten \cite{Witten} conjectured that the generating series of the intersection numbers of 
$\psi$-classes on $\Mbar_{g,n}$ is a tau-function of the KdV hierarchy specified by the puncture equation. 
Witten also made generalizations to include equivalence with multimatrix models in \cite{Witten2, Witten3}. 
The relevant algebraic geometric objects are the moduli spaces of $r$-spin curves.

Witten's work have inspired much progress in the mathematical literature.
His conjecture in \cite{Witten} has been proved by Kontsevich \cite{Kon} using a different kind of matrix models,
now called the Kontsevich matrix models, by Mirzhakhani \cite{Mir} using hyperbolic geometry,
and by Okounkov and Pandhanripande and some other authors using Hurwitz numbers \cite{Oko-Pan, Kaz-Lan, Chen-Li-Liu, Kiem-Liu}. 
His conjecture in \cite{Witten3} has been proved by \cite{Fab-Sha-Zvo}.
His work also forms the foundation of Gromov-Witten theory and FJRW theory \cite{FJR}, Frobenius manifolds \cite{Dub},
 etc. 

Now the geometric theory of topological gravity is fully established,
the phase transitions in this theory have yet to be fully explored.
A first step towards this direction was made by Dijkgraaf and Witten \cite{Dijkgraaf-Witten}. 
They used the topological recursion relations and the puncture equation to derive the string equation in genus zero:
\be 
u = t_0 +\sum_{i=1}^\infty t_i u^i 
\ee
and interpreted it as a Landau-Ginzburg equation. 
The Landau-Ginzburg potential is then 
\be  
W(u) = - \half u^2 +\sum_{i=0}^\infty \frac{t_i}{i+1} u^i.
\ee 
The $k$-th critical point correspond to the particular value $t_1 =1$, $t_k = -1$ while other $t_i$'s are equal to $0$.  
Such an interpretation of the string equation as the Landau-Ginzburg equation
already has some feature of the Wilson paradigm where infinitely many coupling constants are 
considered.
After the renormalization, the bare coupling constants are supposed to become the renormalized coupling constants. 
The partition functions and the correlation functions are a priori formal power series in $t_i$'s,
and in general one does not expect they are convergent. 
In \cite{Itz-Zub},
Itzykson and Zuber introduced formal power series $\{I_n\}_{n \geq 0}$ in $\{t_i\}_{i \geq 0}$.
They showed 
\ben
F_0 &= &\frac{I_0^3}{6}- \sum_{k \geq 0} \frac{I_0^{k+2}}{k+2} \frac{t_k}{k!}
+\half \sum_{k \geq 0}\frac{I_0^{k+1}}{k+1} \sum_{a+b=k} \frac{t_a}{a!}\frac{t_b}{b!},\\ 
F_1 & = &\frac{1}{24}\log \frac{1}{1-I_1},
\een
and make the ansatz
that for $g\geq 2$, $F_g$ is a polynomial in $I_k/(1 - I_1)^{\frac{2k+1}{3}}$.  

In earlier work we interpreted the $I_n$'s as the renormalized $t_n$'s.
We considered the matrix model of size $1\times 1$ and called it the topological 1D gravity \cite{Zhou-1D} (see also \cite{NY}).
In such a model a renormalization procedure produces the $\{I_n\}_{n \geq 0}$ from $\{t_n\}_{n\geq 0}$ 
in the limit.  
We then understood $\{I_n\}_{n \geq 0}$ and $\{t_n\}_{n \geq 0}$ as different coordinates on a big phase 
space and compute the Jacobian matrices of the coordinate changes in \cite{Zhou-1D}. 
In \cite{Zhang-Zhou} we prove the Itzykson-Zuber ansatz.
In this paper we will use these earlier results to study the phase transitions
in two-dimensional topological gravity.
In particular, 
we will show how to give a uniform way to compute the free energy and correlation functions
at all the critical points using the renormalized coupling constants.
This has the advantage that it is clear that all results are essentially polynomial expressions 
of some simple form,
and also it is easy to check the critical exponents predicted by 
the DDK/KPZ formula at all critical points and at all genera.  

We arrange the rest of the paper as follows.
In Section \ref{sec:Big-Phase} we recall the big phase space of the 2D topological gravity
and the multicritical points on it as seen from the Landau-Gingzburg potential.
We also explain how the dilaton equation gives rise to difficulty in reach the multicritical points starting from the Witten-Kontsevich tau-function.
We explain how to compute the free energy and correlation functions in the renormalized coupling constants 
in Section \ref{sec:I}.
This is the key technical Section.
As illustrations we present some concrete computation in the case of pure gravity (the $k=2$ case) in  Section \ref{sec:k=2}
and the case of Yang-Lee edge singularity coupled with gravity (the $k=3$ case) in Section \ref{sec:k=3}.
We determine the critical exponents of the free energy functions and correlation function in each genus using the renormalized coordinates
in Section \ref{sec:Exponents}.

\section{Critical Points on the Big Phase Space of 2D Topological Gravity}
\label{sec:Big-Phase}

In this Section we recall the big phase space of two-dimensional gravity and the critical points on it.
We will review the special roles played by the puncture operator $\tau_0$ and the dilaton operator $\tau_1$.

\subsection{Big phase space}

Let us begin by recalling the Witten-Kontsevich tau-function 
which is the partition function of the type $(1,2)$-theory in two-dimensional topological gravity.
The relevant genus $g$, $n$-point correlators are defined by the intersection numbers:
\be 
\corr{\tau_{a_1}\cdots \tau_{a_n}}_g:
= \int_{\Mbar_{g,n}} \psi_1^{a_1} \cdots \psi_n^{a_n}.
\ee
For degree reasons, 
the following selection rule is satisfied: $\corr{\tau_{a_1}\cdots \tau_{a_n}}_g\neq 0$
only if
\be  
a_1 +\cdots +a_n = 3g-3+n.
\ee
From this rule it is clear that if all $a_i$’s are bigger than $1$,
then  for each fixed genus $g$,
there are only finitely many nonzero correlators to compute.
So in order to get finite expressions,
we need some recursion relations that remove all the $\tau_0$'s and $\tau_1$'s from the correlators.
This can be achieved by the puncture equation and the dilaton equation to be recalled below.

The genus $g$ free energy is defined by:
\be 
F_g(t_0, t_1, \dots) 
: = \sum_{n: n> 2-2g} \frac{t_{a_1}\cdots t_{a_n}}{n!} \corr{\tau_{a_1}\cdots \tau_{a_n}}_g.
\ee
The all genera free energy is defined by
\be 
F(t_0, t_2, \dots;\lambda) = \sum_{g\geq 0}\lambda^{2g-2} F_g(t_0, t_1, \dots),
\ee
and the partition function 
\be 
Z(t_0, t_1, \dots; \lambda):=\exp F(t_0, t_2, \dots;\lambda)
\ee
is called the Witten-Kontsevich tau-function.

By the big phase space  we mean the infinite-dimensional vector space 
with linear coordinates $\{t_n\}_{n \geq 0}$.
Because $F_g$, $F$ and $Z$ are just defined as formal power series in the $t_n$'s,
a priori, they are understood to make sense only near the origin where all $t_n$'s are equal to $0$.  
Our purpose is to understand the big phase space as a moduli spaces of theories so that
the Witten-Kontsevich tau-function corresponds to the  theory with all $t_n = 0$,
and all other theories of type $(2k-1,2)$ lie in hypersurface with $t_1 =1$.

\subsection{KdV hierarchy}

The specific heat is defined as the second derivative of the free energy:
\be
u(\bt; \lambda) = \sum_{g \geq 0} \lambda^{2g} u_g(\bt)
= \sum_{g \geq 0} \lambda^{2g} \frac{\pd^2F_g(\bt)}{\pd t_0^2}.
\ee
According to Witten Conjecture/Kontsevich Theorem \cite{Witten, Kon},
$u$ should satisfy the KdV hierarchy
\be
\frac{\pd u}{\pd t_n} = \pd_{x} R_{n+1}[u],
\ee
for the sequence of Gelfand-Dickey differential polynomials $R_n[u]$.
define by the Lenard recursion relations:
\be \label{eqn:Lenard}
\begin{split}
& R_0 = 1, \\
& \pd_x R_{n+1} = \frac{1}{2n+1} \biggl(\pd_xu \cdot R_n
+ 2u \cdot \pd_x R_n + \frac{\lambda^2}{ 4} \pd_x^3R_n \biggr).
\end{split}
\ee
For example,
\ben
\pd_x R_1 & = & \pd_x u, \\
R_1 & = & u, \\
\pd_x R_2 & = & u \cdot \pd_x u + \frac{\lambda^2}{12} \pd_x^3 u, \\
R_2 & = & \frac{1}{2} u^2 + \frac{\lambda^2}{12} \pd_x^2u,  \\
\pd_x R_3 & = & \frac{1}{2} u^2\cdot \pd_x u + \frac{\lambda^2}{12} u \cdot \pd_x^3u
+ \frac{\lambda^2}{6}\pd_x u \cdot \pd_x^2u + \frac{\lambda^4}{240}\pd_x^5u, \\
R_3 & = & \frac{1}{6} u^3 + \frac{\lambda^2}{12} u \cdot \pd_x^2 u
+ \frac{\lambda^2}{24} (\pd_x u)^2 + \frac{\lambda^4}{240} \pd_x^4 u.
\een

\subsection{The puncture equation}

The coupling constant $t_0$ corresponds to the {\em puncture operator}.
The recursion relation satisfied by the insertion of the puncture operator is  
the puncture equation \cite{Witten}:
\be  \label{eqn:Puncture}
\frac{\pd F}{\pd t_0} = \sum_{n \geq 0} t_{n+1} \frac{\pd F}{\pd t_n} +  \frac{t_0^2}{2}.
\ee
In terms of correlators,
\be 
\corr{\tau_0\tau_{a_1}\cdots \tau_{a_n}}_g
= \sum_{i=1}^n  \corr{\tau_{a_1} \cdots \tau_{a_i-1} \cdots \tau_{a_n}}_g.
\ee
From this one sees that one can reduce correlators without $\tau_0$.
One well-known application of the puncture equation is to the removal of 
all the insertions of the puncture operator from the correlator.

\subsection{The string equation}
Let us recall another well-known application of the puncture equation: 
The derivation of the string equation from the puncture equation \cite{DVV}.
Taking $\frac{\pd^2}{\pd t_0^2}$ twice on both sides of \eqref{eqn:Puncture} twice
and use the KdV hierarchy for $u$, 
one gets:
\be 
\frac{\pd u}{\pd t_0} = \sum_{n \geq 0} t_{n+1} \frac{\pd u}{\pd t_n} +  1 
=  \sum_{n \geq 0} t_{n+1}\pd_x R_{n+1}[u] + 1.
\ee
The string equation was then obtained by integrating with respect to $t_0$ once:
\be \label{eqn:String}
0 = \sum_{k=1}^\infty ( t_k - \delta_{k,1}) R_k + x.
\ee

\subsection{Landau-Ginzburg equation}

Expanding both sides of the string equation as series in $\lambda$,
one gets by comparing the leading terms the following equation:
\be \label{eqn:Landau-Ginzburg-0}
u_{0} = \sum_{k=1}^\infty \frac{1}{k!} t_k \cdot u_{0}^k + x.
\ee
We will call this equation  the {\em Landau-Ginzburg equation}.
It can be rewritten as follows.
\be \label{eqn:Landau-Ginzburg}
\sum_{k=0}^\infty \frac{1}{k!} (t_k - \delta_{k,1}) \cdot u_{0}^k   = 0.
\ee
Here $u_0$ is the genus zero part of $u$:
\be
u_0 = \frac{\pd F_0}{\pd t_0^2}.
\ee

\subsection{Critical points and multi-critical points}

The idea of Landau and Ginzburg was to minimize some action of some mean field theory.
For the genus zero part it is straightforward to consider:
\be
G(v; \bt) = \sum_{k=0}^\infty \frac{1}{(k+1)!} (t_k - \delta_{k,1}) \cdot v^{k+1}.
\ee
Then the Landau-Ginzburg equation is just the equation:
\be
\frac{\pd}{\pd v} G(v; \bt) = 0.
\ee

Now note
\be
\frac{\pd^2}{\pd v^2} G(v; \bt)
= - 1 + \sum_{k=0}^\infty \frac{1}{k!} t_{k+1} \cdot v^{k}.
\ee
Therefore,
when $t_1= t_2 = \cdots = 0$,
$$\frac{\pd^2}{\pd v^2} G(v; \bt)|_{t_k =0, k\geq 1} = - 1 \neq 0.$$
However,
when $t_1 = 1$, $t_2 = t_3 = \cdots =0$,
$$\frac{\pd^2}{\pd v^2} G(v; \bt)|_{t_1 =1, t_k =0, k\geq 2} = 0.$$
But note
\be
\frac{\pd^3}{\pd v^3} G(v; \bt)
= t_2 + \sum_{k=1}^\infty \frac{1}{k!} t_{k+2} \cdot v^{k}.
\ee
Therefore,
when $t_1 =1$, $t_2 \neq 0$,
one has
\be
\frac{\pd}{\pd v} G(v; \bt) = \frac{\pd^2}{\pd v^2} G(v; \bt) =0, \;\; \frac{\pd^3}{\pd v^3} G(v; \bt) \neq 0,
\ee
Continue this way of thinking,
take $t_1=1$, $t_k=0$ for $k \geq 2$ except for $t_m \neq 0$.
Then
\be
\frac{\pd^k}{\pd v^k} G(v; \bt) =0, \;\; 1\leq k \leq m-1, \;\;\;
\frac{\pd^m}{\pd v^m} G(v; \bt) \neq 0.
\ee
In this situation, physicists say $G(v;\bt)$ as a function of $v$
acquires a multi-critical point when such parameters are chosen.

The $n$-th critical point is where one has
\be
t_n \neq 0, \;\;\;\;\; t_k = 0, \;\; k \neq n.
\ee
These points are where the $(2n-1, 2)$-theories arise.

\subsection{Dilaton equation}
The coupling constant $t_1$ corresponds to the dilaton operator $\tau_1$.
The recursion relation satisfied by the insertion of $\tau_1$ is called the {\em dilaton equation}:
\be 
\corr{\tau_1\tau_{a_1}\cdots \tau_{a_n}}_g
= \sum_{i=1}^n \frac{2a_i+1}{3} \corr{\tau_{a_1} \cdots \tau_{a_n}}_g.
\ee
Using this equation,
one can reduce the calculations of all correlators $\corr{\tau_{a_1}\cdots \tau_{a_n}}_g$
to such correlation so that none of the $a_i$'s is equal to $1$.
As a consequence,
one can show that $F_g(t_0, t_1, \dots)$ takes the following form:
\be 
\begin{split}
F_g(t_0, t_1, \dots)
& = \sum_{\substack{n: n> 2-2g \\a_1,\dots, a_n \neq 1}} 
\frac{1}{n!} \corr{\tau_{a_1}\cdots \tau_{a_n}}_g \prod_{i=1}^n \frac{t_{a_i}}{(1-t_1)^{(2a_i+1)/3}} \\
& =  \sum_{\substack{n: n> 2-2g \\a_1,\dots, a_n \neq 1}} 
\frac{t_{a_1}\cdots t_{a_n}}{n!(1-t_1)^{2g-2+n}} \corr{\tau_{a_1}\cdots \tau_{a_n}}_g.
\end{split}
\ee
From this it seems that we cannot take $t_1 =1$ at all. 
But we are working in a perturbative setting,
and we are working with formal power series with infinitely many other coupling constants $t_n$'s.
We will see later that it is possible to change to the renormalized constants to get finite expressions,
and then it is actually possible to reach the hyperplane $t_1 =1$ where all the multicritical points lie.    
 
\section{Free Energy  and $N$-Point Functions in  Renormalized Coupling Constants}
\label{sec:I}

In this Section we will see that using some formal power series $\{I_n\}_{n \geq 0}$ introduced
by Itzykson-Zuber \cite{Itz-Zub},
it is possible to get finite expressions for the genus $g$ free energy for $g \geq 2$
and $n$-point functions in genus $g$ when $2g-2+n> 0$.
These series have been  interpreted as renormalized coupling constants in earlier work \cite{Zhou-1D, Zhang-Zhou}.

\subsection{Renormalized coupling constants}

Using $\{I_n\}_{n \geq 0}$ introduced by Itzykson-Zuber \cite{Itz-Zub},
it is possible to obtain finite expressions for $F_g$'s.
They are defined by:
\begin{equation}\label{Def:In}
I_k(t):=\sum_{p\geqslant0}\frac{I_{0}^p}{p!}t_{p+k}, \ \ \ k\geqslant 0,
\end{equation}
where $I_0$ is a formal power series solution to the string equation in genus zero:
\be 
I_0 = \sum_{n=0}^\infty \frac{t_n}{n!} I_0^n.
\ee

\subsection{Free energy in renormalized coupling constants}
With these definitions, $F_g$'s  take the following form \cite{Itz-Zub, Zhang-Zhou}:
\begin{align} \label{eqn:F0}
F_0=\frac{1}{6}I_0^3-\sum_{p\geqslant0}\frac{(-1)^pI_0^{p+2}}{(p+2)!} I_p +\frac{1}{2}\sum_{p,q\geqslant0}\frac{(-1)^{p+q}I_0^{p+q+1}}{(p+q+1)q!p!}I_pI_{q}
\end{align}

\begin{equation} \label{eqn:F1}
F_{1}=\sum_{n=1}^{\infty}\frac{1}{24n}I_1^n=\frac{1}{24}\log{\frac{1}{1-I_1}}
\end{equation}

\begin{equation}\label{eqn:Fg}
F_g(t)=\sum_{\sum_{2\leqslant k\leqslant 3g-2}(k-1)l_k=3g-3}\left<\tau_2^{l_2}\tau_{3}^{l_{3}}\cdots\tau_{3g-2}^{l_{3g-2}} \right>_g \prod_{j=2}^{3g-2}\frac{1}{l_j!}\left(\frac{I_j}{(1-I_1)^{\frac{2j+1}{3}}}\right)^{l_j}
\end{equation}
for $g\geqslant2$.
For example, 
\begin{align}
F_{2}=&\frac{7}{1440}\tilde{I}_2^3+\frac{29}{5760}\tilde{I}_2\tilde{I}_3+\frac{1}{1152}\tilde{I}_4\\
F_{3}=&\frac{245}{20736}\tilde{I}_2^6+\frac{193}{6912}\tilde{I}_2^4\tilde{I}_3 +\frac{205}{13824}\tilde{I}_2^2\tilde{I}_3^3 +\frac{583}{580608}\tilde{I}_3^3 +\frac{53}{6912}\tilde{I}_2^3\tilde{I}_4+\frac{1121}{241920}\tilde{I}_2\tilde{I}_3\tilde{I}_4 \nonumber\\ &+\frac{607}{2903040}\tilde{I}_4^2+\frac{17}{11520}\tilde{I}_2^2\tilde{I}_5 +\frac{503}{1451520}\tilde{I}_3\tilde{I}_5 +\frac{77}{414720}\tilde{I}_2\tilde{I}_6+\frac{1}{82944}\tilde{I}_7\\
F_{4}=&\frac{259553}{2488320}\tilde{I}_2^9+\frac{475181}{1244160}\tilde{I}_2^7\tilde{I}_3 +\frac{145693}{331776}\tilde{I}_2^5\tilde{I}_3^2 +\frac{43201}{248832}\tilde{I}_2^3\tilde{I}_3^3 +\frac{134233}{7962624}\tilde{I}_2\tilde{I}_3^4+\frac{14147}{124416}\tilde{I}_2^6\tilde{I}_4 \nonumber\\ &+\frac{83851}{414720}\tilde{I}_2^4\tilde{I}_3\tilde{I}_4+\frac{26017}{331776}\tilde{I}_2^2\tilde{I}_3^2\tilde{I}_4 +\frac{185251}{49766400}\tilde{I}_3^3\tilde{I}_4 +\frac{5609}{276480}\tilde{I}_2^3\tilde{I}_4^2 +\frac{177}{20480}\tilde{I}_3\tilde{I}_4^2 \nonumber\\ &+\frac{175}{995328}\tilde{I}_4^3+\frac{21329}{829440}\tilde{I}_2^5\tilde{I}_5 +\frac{13783}{414720}\tilde{I}_2^3\tilde{I}_3\tilde{I}_5 +\frac{1837}{259200}\tilde{I}_2\tilde{I}_3^2\tilde{I}_5 +\frac{7597}{1382400}\tilde{I}_2^2\tilde{I}_4\tilde{I}_5\nonumber\\
&+\frac{719}{829440}\tilde{I}_3\tilde{I}_4\tilde{I}_5 +\frac{533}{1935360}\tilde{I}_2\tilde{I}_5^2 +\frac{2471}{552960}\tilde{I}_2^4\tilde{I}_6+\frac{7897}{2073600}\tilde{I}_2^2\tilde{I}_3\tilde{I}_6 +\frac{1997}{6635520}\tilde{I}_3^2\tilde{I}_6\nonumber\\
&+\frac{1081}{2322432}\tilde{I}_2\tilde{I}_4\tilde{I}_6 +\frac{487}{18579456}\tilde{I}_5\tilde{I}_6 +\frac{4907}{8294400}\tilde{I}_2^3\tilde{I}_7+\frac{16243}{58060800}\tilde{I}_2\tilde{I}_3\tilde{I}_7 +\frac{1781}{92897280}\tilde{I}_4\tilde{I}_7\nonumber\\
&+\frac{53}{921600}\tilde{I}_2^2\tilde{I}_8 +\frac{947}{92897280}\tilde{I}_3\tilde{I}_8+\frac{149}{39813120}\tilde{I}_2\tilde{I}_9 +\frac{1}{7962624}\tilde{I}_{10}
\end{align}
where
\begin{equation}
\tilde{I}_j=\frac{I_j}{(1-I_1)^{(2j+1)/3}}
\end{equation}
In particular, when $g >1$,
$F_g$ is a weighted homogeneous polynomial in $\tilde{I}_n$'s of degree $3g-3$ if we assign
\be 
\deg \tilde{I}_n = n-1.
\ee
Even though $F_0$ and $F_1$ are not polynomials,
but they are also weighted homogeneous of degree $-3$ and $0$ respectively if we assign
\be 
\deg I_n = n-1.
\ee

\subsection{Loop operator in $I$-coordinates}

The goal of this Section is to show that the correlation functions can also be expressed by the $I$-coordinates.
For this purpose we need to express the loop operator in terms of the $I$-coordinates.
Let us first recall the following result proved in \cite{Zhou-1D}:

\begin{prop}
The vector fields $\{\frac{\pd}{\pd t_k}\}$  can be expressed in terms of
the vector fields  $\{\frac{\pd}{\pd I_l}\}$ as follows:
\be \label{eqn:diff-tk}
\frac{\pd}{\pd t_k}
= \frac{1}{1-I_1}\frac{I_0^k}{k!}  \frac{\pd}{\pd I_0} +
\frac{I_0^k}{k!} \sum_{l \geq 1} \frac{I_{l+1}}{1-I_1} \frac{\pd}{\pd I_l}
+ \sum_{1 \leq l \leq k} \frac{I_0^{k-l}}{(k-l)!} \frac{\pd}{\pd I_l}.
\ee
\end{prop}

For example,
\bea
&& \frac{\pd}{\pd t_0} 
= \frac{1}{1-I_1} \frac{\pd}{\pd I_0}
+ \sum_{l \geq 1} \frac{I_{l+1}}{1-I_1} \frac{\pd}{\pd I_l}, \label{eqn:Pd-t0-in-I} \\
&& \frac{\pd}{\pd t_1} 
=  \frac{I_0}{1-I_1} \frac{\pd}{\pd I_0} + \biggl(\frac{I_2I_0}{1-I_1} + 1\biggr) \frac{\pd}{\pd I_1}
+ \sum_{l \geq 2} \frac{I_{l+1}I_0}{1-I_1} \frac{\pd}{\pd I_l}.
\eea
Recall the loop operator is defined by:
\be 
\nabla(x)
= \sum_{n \geq 0} \frac{(2n+1)!!}{x^{n+\frac{3}{2}} }\frac{\pd}{\pd t_n}.
\ee

\begin{thm}
The loop operator can be expressed in terms of the $I$-coordinates as follows: 
\be \label{eqn:nabla(x)}
\begin{split}
\nabla(x)  
 = &    \frac{1}{(x-2I_0)^{3/2}} \cdot \frac{1}{1-I_1}   \frac{\pd}{\pd I_0} \\
 + &
  \sum_{l \geq 1} \biggl(\frac{1}{(x-2I_0)^{3/2}} \cdot  \frac{I_{l+1}}{1-I_1} 
  +\frac{(2l+1)!!}{(x-2I_0)^{(2l+3)/2}} \biggr)\frac{\pd}{\pd I_l}.
\end{split}
\ee
\end{thm}

\begin{proof}
We have:
\ben 
\nabla(x) & = &  \sum_{n \geq 0} \frac{(2n+1)!!}{x^{n+\frac{3}{2}} }  \biggl( \frac{1}{1-I_1}\frac{I_0^n}{n!}  \frac{\pd}{\pd I_0} +
\frac{I_0^n}{n!} \sum_{l \geq 1} \frac{I_{l+1}}{1-I_1} \frac{\pd}{\pd I_l}
+ \sum_{1 \leq l \leq n} \frac{I_0^{n-l}}{(n-l)!} \frac{\pd}{\pd I_l} \biggr) \\
& = & \sum_{n \geq 0} \frac{(2n+1)!!}{x^{n+\frac{3}{2}} } \frac{I_0^n}{n!} \cdot  \biggl( \frac{1}{1-I_1}   \frac{\pd}{\pd I_0} 
+ \sum_{l \geq 1} \frac{I_{l+1}}{1-I_1} \frac{\pd}{\pd I_l} \biggr) \\
& + &   \sum_{l=1}^\infty \frac{(2l+1)!!}{x^{l+\frac{3}{2}}} (1-\frac{2I_0}{x})^{-(2l+3)/2} \frac{\pd}{\pd I_l}  \\
& = &  \frac{1}{x^{\frac{3}{2}}}  \frac{1}{(1-\frac{2I_0}{x})^{3/2}} \biggl( \frac{1}{1-I_1}   \frac{\pd}{\pd I_0} +
  \sum_{l \geq 1} \frac{I_{l+1}}{1-I_1} \frac{\pd}{\pd I_l}  \biggr) \\
& + &   \sum_{l=1}^\infty \frac{(2l+1)!!}{(x-2I_0)^{(2l+3)/2}} \frac{\pd}{\pd I_l}.
\een
In the above we have some combinatorial identities as follows.
First note
\ben
\sum_{n=0}^\infty \frac{(2n+1)!!}{n!} t^n 
& = &  \sum_{n=0}^\infty \frac{(2n+1)!}{n!n!} (\frac{t}{2})^n = (1-2t)^{-3/2}.
\een
Take $\frac{d}{dt}$ on both sides:
\ben
\sum_{n=0}^\infty \frac{(2n+3)!!}{n!} t^n 
& = &  3  \cdot (1-2t)^{-5/2}.
\een
Repeat this for several times:
\ben
\sum_{n=0}^\infty \frac{(2n+2l+1)!!}{n!} t^n 
& = & (2l+1)!!   \cdot (1-2t)^{-(2l+3)/2}.
\een
 \end{proof}

\subsection{Correlation functions in genus $\geq 1$ in $I$-coordinates}

Recall 
\be
F_{1}=\frac{1}{24}\log{\frac{1}{1-I_1}}.
\ee
By applying the loop operator repeatedly,
one can get all the $n$-point functions in genus $1$, for example
\be 
\nabla(x) F_1 
= \frac{1}{24} \frac{1}{(x-2I_0)^{3/2}} \cdot \frac{I_2}{(1-I_1)^2} + \frac{1}{8} \frac{1}{(x-2I_0)^{5/2}}\cdot \frac{1}{1-I_1},
\ee
\ben
&& \nabla(x_1)\nabla(x_2)F_1 \\
& = & \frac{5}{8(x_1-2I_0)^{\frac{3}{2}} (x_2-2I_0)^{\frac{7}{2}}(1-I_1)^2}
+  \frac{3}{8(x_1-2I_0)^{\frac{5}{2}} (x_2-2I_0)^{\frac{5}{2}}(1-I_1)^2} \\
& + & \frac{5}{8(x_1-2 I_0)^{\frac{7}{2}} (x_2-2 I_0)^{\frac{3}{2}} (1-I_1)}  \\
& + & \frac{I_2}{4(x_1-2I_0)^{\frac{3}{2}}(x_2-2I_0)^{\frac{5}{2}} (1-I_1)^3}
+   \frac{I_2}{4(x_1-2I_0)^{\frac{5}{2}}(x_2-2I_0)^{\frac{3}{2}} (1-I_1)^3} \\
&+ & \frac{I_3}{24(x_1-2 I_0)^{\frac{3}{2}} (x_2-2 I_0)^{\frac{3}{2}}  (1-I_1)^3} 
+ \frac{I_2^2}{12(x_1-2 I_0)^{\frac{3}{2}} (x_2-2 I_0)^{\frac{3}{2}}  (1-I_1)^4}.
\een
Note these are polynomial expressions in $(x_j-2I_0)^{-1/2}$, $\frac{1}{1-I_1}$ and $I_k$, $k \geq 2$.
When $g > 1$, 
$F_g$ is already a polynomial in $\frac{1}{1-I_1}$ and $I_k$ for $k \geq 2$.
For example,
\be
F_{2}=\frac{7}{1440}\frac{I_2^3}{(1-I_1)^7}+\frac{29}{5760}\frac{I_2I_3}{(1-I_1)^4}+\frac{1}{1152}\frac{I_4}{(1-I_1)^3}.
\ee
Then the $n$-point correlation functions are polynomials in $(x_j-2I_0)^{-1/2}$, $\frac{1}{1-I_1}$ and $I_k$, $k \geq 2$.
For example,
\ben
\nabla(x)F_2 & = &  \frac{7 I_2^4}{288(1-I_1)^7 (x-2I_0)^{\frac{3}{2}} }
+  \frac{5I_2^2I_3}{144(1-I_1)^6  (x-2I_0)^{\frac{3}{2}}} \\
&+ &   \frac{11I_2I_4}{1440(1-I_1)^5  (x-2I_0)^{\frac{3}{2}} }
+ \frac{29I_3^2}{5760(1-I_1)^5 (x-2I_0)^{\frac{3}{2}}}  \\ 
&+ &   \frac{I_5}{1152(1-I_1)^4(x-2I_0)^{\frac{3}{2}}}   
+  \frac{7I_2^3}{96(1-I_1)^6 (x-2I_0)^{\frac{5}{2}} }    \\
& + &   \frac{29I_2I_3}{480 (1-I_1)^5(x-2I_0)^{\frac{5}{2}} }  
+  \frac{I_4}{128 (1-I_1)^4 (x-2I_0)^{\frac{5}{2}} } \\
& + & \frac{7I_2^2}{32(1-I_1)^5(x-2I_0)^{\frac{7}{2}}}    +  \frac{29I_3}{384(1-I_1)^4(x-2I_0)^{\frac{7}{2}}} \\
&  + & \frac{203I_2}{384(1-I_1)^4(x-2I_0)^{\frac{9}{2}}}  
+ \frac{105}{128(1-I_1)^3(x-2I_0)^{\frac{11}{2}}} .
\een
Such computations are very easy to automate with some computer algebra systems such as Maple. 

Assign the following grading:
\begin{align}
\deg (x-2I_0)^{-(2k+1)/2} & = \frac{2k+1}{3}, & \deg \frac{1}{1-I_1} & = -1, \\ \deg I_k & = \frac{2k+1}{3}, k > 1. 
\end{align}
Then from 
\ben
\nabla(x) ((x'-2I_0)^{-\frac{2k+1}{2}}) & = & - \frac{2k+1}{2} \cdot (x-2I_0)^{-\frac{3}{2}} \cdot \frac{1}{1-I_1} \cdot  (x'-2I_0)^{-\frac{2k+3}{3}}, \\
\nabla(x)(\frac{1}{1-I_1}) & = & (x-2I_0)^{-\frac{3}{2}} \frac{I_2}{(1-I_1)^3} + \frac{3}{(x-2I_0)^{\frac{5}{2}}}, \\
\nabla(x) (I_l) & = & \frac{1}{(x-2I_0)^{3/2}} \cdot  \frac{I_{l+1}}{1-I_1} 
  +\frac{(2l+1)!!}{(x-2I_0)^{(2l+3)/2}}, l\geq 2,
\een
one sees that $\nabla(x)$ has degree $\frac{2}{3}$.
It follows that $\nabla(x_1) \cdots \nabla(x_n)F_g$ is a weighted homogeneous polynomial of degree $\frac{2}{3}n$ 
in $(x_j-2I_0)^{-1/2}$, $\frac{1}{1-I_1}$ and $I_k$, $k \geq 2$.
Alternatively,
 if we introduce
\begin{align}
\hat{x} & =  \frac{1}{(x-2I_0)^{\half}(1-I_1)^{\frac{1}{3}}}, & \hat{I}_l & =  \frac{I_l}{(1-I_1)^{\frac{2l+1}{3}}},
\end{align}
then $(1-I_1)^{-\frac{2n}{3}} \nabla(x_1) \cdots \nabla(x_n)F_g$ is a polynomial in $\hat{x}_1$, $\dots$, $\hat{x}_n$
and $\hat{I}_l$, $l > 1$,
such that the exponent of $\hat{x}_i$ in each monomial term is an odd integer $\geq 3$.
For example, 
\ben 
\nabla(x) F_1 
& = & (1-I_1)^{\frac{2}{3}} \cdot \biggl[\frac{1}{24} 
\biggl(\frac{1}{(x-2I_0)^{\half}(1-I_1)^{\frac{1}{3}}}\biggr)^3 \cdot \frac{I_2}{(1-I_1)^{\frac{5}{3}}} \\
& + & \frac{1}{8} \biggl(\frac{1}{(x-2I_0)^{\half}(1-I_1)^{\frac{1}{3}} }\biggr)^5 \biggr].
\een

\subsection{Correlation functions in genus $0$ in $I$-coordinates}

Recall 
\ben
F_0& = & \frac{1}{6}I_0^3-\sum_{p\geqslant0}\frac{(-1)^pI_0^{p+2}}{(p+2)!} I_p +\frac{1}{2}\sum_{p,q\geqslant0}\frac{(-1)^{p+q}I_0^{p+q+1}}{(p+q+1)q!p!}I_pI_{q} \\
& = &  \frac{1}{6}I_0^3 
+ \sum_{p\geqslant1}\frac{(-1)^pI_0^{p+2}}{(p+2)\cdot p!} I_p
 +\frac{1}{2}\sum_{p,q\geqslant1}\frac{(-1)^{p+q}I_0^{p+q+1}}{(p+q+1)q!p!}I_pI_{q} \\
\een
This is no longer a finite expression.
The direct computation of the one-point function in genus $0$ leads to complicated expressions:
\ben
\nabla(x) F_0 & = &  \frac{1}{(x-2I_0)^{3/2}} \cdot \frac{1}{1-I_1}  \biggl( \frac{I_0^2}{2} +\sum_{p\geq  1} \frac{(-1)^pI_0^{p+1}}{p!} I_p   \\
& + & \frac{1}{2}\sum_{p,q\geq1}\frac{(-1)^{p+q}I_0^{p+q}}{p!q!}I_pI_{q} \biggr) \\ 
& + &
  \sum_{l \geq 1} \biggl(\frac{1}{(x-2I_0)^{3/2}} \cdot  \frac{I_{l+1}}{1-I_1} 
  +\frac{(2l+1)!!}{(x-2I_0)^{(2l+3)/2}} \biggr)  \\
&& \cdot \biggl( \frac{(-1)^l}{(l+2) \cdot l!} I_0^{l+2}I_l 
+\sum_{q\geq 1}\frac{(-1)^{l+q}I_0^{l+q+1}}{(l+q+1) \cdot l!q!}I_lI_{q} \biggr).
\een
A trick has been used to obtain the following result in \cite{Zhang-Zhou}:
If one sets
\be
y(x) :=   x^{\frac{1}{2}}-\sum_{n=0}^{\infty}t_n  \frac{1}{(2n-1)!!}x^{n-\frac{1}{2}} 
- \sum_{n=0}^{\infty} \frac{\partial F_0}{\partial t_n} (2n+1)!! x^{-n-\frac{3}{2}}, 
\ee
then one has:
\be 
y = - 
\sum_{n=1}^{\infty} \frac{I_n-\delta_{n,1}}{(2n-1)!!} (x-2I_0)^{n-\frac{1}{2}}.
\ee

\ben 
&& - \nabla(x_1) y(x_2) \\
& = & -
 \frac{1}{(x_1-2I_0)^{3/2}} \cdot \frac{1}{1-I_1} \sum_{n=1}^{\infty} \frac{I_n-\delta_{n,1}}{(2n-1)!!} \cdot (2n-1)\cdot (x_2-2I_0)^{n-\frac{3}{2}} \\
&+&\sum_{n=1}^{\infty} \frac{(x_2-2I_0)^{n-\frac{1}{2}}}{(2n-1)!!}\cdot 
\biggl(\frac{1}{(x_1-2I_0)^{3/2}} \cdot  \frac{I_{n+1}}{1-I_1} 
  +\frac{(2n+1)!!}{(x_1-2I_0)^{(2n+3)/2}} \biggr)\\ 
&=&\sum_{n=0}^{\infty} (2n+1) \frac{(x_2-2I_0)^{n-\frac{1}{2}}}{(x_1-2I_0)^{n+\frac{3}{2}}} \\ 
& = &  \frac{1}{(x_1-2I_0)^{\frac{1}{2}} (x_2-2I_0)^{\frac{1}{2}}} \cdot 
\frac{x_1+ x_2-4I_0}{(x_1-x_2)^2}.
\een
Note we also have
\ben
&& \nabla(x_1) \biggl(x_2^{\frac{1}{2}}-\sum_{n=0}^{\infty}t_n  \frac{1}{(2n-1)!!}x_2^{n-\frac{1}{2}} \biggr) \\
& = &  - \sum_{n}^\infty (2n+1) \frac{x_2^{n-\half}}{x_1^{n+\frac{3}{2}}} \\
& = & - \frac{x_2^{-\frac{1}{2}}}{x_1^{\frac{3}{2}}} \frac{1+ \frac{x_2}{x_1}}{(1-\frac{x_2}{x_1})^2} 
= - \frac{x_1+x_2}{x_1^{\frac{1}{2}}x_2^{\frac{1}{2}}(x_1-x_2)^2}.
\een
It follows that we have
\be 
\begin{split}
& \nabla(x_1)\nabla(x_2)F_0 \\
= &\frac{1}{(x_1-2I_0)^{\frac{1}{2}} (x_2-2I_0)^{\frac{1}{2}}} \cdot 
\frac{x_1+ x_2-4I_0}{(x_1-x_2)^2}
- \frac{x_1+x_2}{x_1^{\frac{1}{2}}x_2^{\frac{1}{2}}(x_1-x_2)^2}.
\end{split}
\ee
By applying loop operator again we get: 
\be 
\begin{split}
& \nabla(x_1)\nabla(x_2)\nabla(x_3)F_0 \\
= &\frac{1}{(x_1-2I_0)^{\frac{3}{2}} (x_1-2I_0)^{\frac{3}{2}}(x_1-2I_0)^{\frac{3}{2}}(1-I_1)}.
\end{split}
\ee
As in the $g> 0$ case,  if we introduce
\begin{align*}
\hat{x} & =  \frac{1}{(x-2I_0)^{\half}(1-I_1)^{\frac{1}{3}}}, & \hat{I}_l & =  \frac{I_l}{(1-I_1)^{\frac{2l+1}{3}}},
\end{align*}
then $(1-I_1)^{-\frac{2n}{3}} \nabla(x_1) \cdots \nabla(x_n)F_0$ is a polynomial in $\hat{x}_1$, $\dots$, $\hat{x}_n$
and $\hat{I}_l$, $l > 1$ when $n \geq 3$.
For example, 
\be
\begin{split}
& \nabla(x_1) \cdots \nabla(x_4)F_0 \\
 = & \frac{I_2}{(x_1-2I_0)^{\frac{3}{2}} (x_1-2I_0)^{\frac{3}{2}}(x_1-2I_0)^{\frac{3}{2}}(x_4-2I_0)^{\frac{3}{2}} (1-I_1)^3} \\
 + &  \frac{3}{(x_1-2I_0)^{\frac{3}{2}} (x_1-2I_0)^{\frac{3}{2}}(x_1-2I_0)^{\frac{3}{2}} (x_4-2I_0)^{\frac{3}{2}}(1-I_1)^2} \\
& \biggl(\frac{1}{x_1-2I_0} + \cdots + \frac{1}{x_4-2I_0} \biggr).
\end{split}
\ee

We summarize the above discussions for all genera in the following:

\begin{thm}
For $g\geq 0$, $n \geq 1$,
when $2g-2+n> 0$, the genus $g$, $n$-point correlation functions
of the two-dimensional topological gravity,
$(1-I_1)^{-\frac{2n}{3}}\nabla_{x_1}\cdots \nabla_{x_n}F_g$,
are weighted homogeneous polynomials in $\hat{x}_1, \dots, \hat{x}_n$ and $\hat{I}_l$ of degree $6g-6+5n$, 
if one assigns 
\begin{align}
\deg \hat{x}_n & = 1, & \deg \hat{I}_l & = 2(l-1).
\end{align} 
\end{thm}
 
 The weighted homogeneity can be checked by induction.
 The following are some examples: 
\ben
&& (1-I_1)^{-2} \nabla(x_1)\nabla(x_2)\nabla(x_3)F_0 = \hat{x}_1^3\hat{x}_2^3\hat{x}_3^3, \\
&&  (1-I_1)^{-\frac{8}{3}} \nabla(x_1)\cdots\nabla(x_4)F_0 = \hat{x}_1^3\cdots\hat{x}_4^3\hat{I}_2
+ 3\hat{x}_1^3\cdots\hat{x}_4^3 \sum_{i=1}^4 \hat{x}_i^2, \\
&& (1-I_1)^{-\frac{2}{3}}\nabla(x)F_1 = \frac{1}{24} \hat{x}^3\hat{I}_2 +\frac{1}{8}\hat{x}^5, \\
&& (1-I_1)^{-\frac{4}{3}}\nabla(x_1)\nabla(x_2)F_1 = \frac{5}{8}\hat{x}_1^3 \hat{x}_2^7
+ \frac{3}{8}\hat{x}_1^5\hat{x}_2^5+\frac{5}{8} \hat{x}_1^7\hat{x}_2^3 \\
&&\quad \quad + \frac{1}{4}\hat{x}_1^3\hat{x}_2^5\hat{I}_2+ \frac{1}{4}\hat{x}_1^5\hat{x}_2^3\hat{I}_2
+ \frac{1}{24} \hat{x}_1^3\hat{x}_2^3\hat{I}_3 +\frac{1}{12}\hat{x}_1^3\hat{x}_2^3\hat{I}_2^2, \\
&& (1-I_1)^{-\frac{2}{3}} \nabla(x)F_2  =  \frac{7}{288} \hat{x}^3 \hat{I}_2^4
+  \frac{5}{144} \hat{x}^3 \hat{I}_2^2\hat{I}_3 
+    \frac{11}{1440} \hat{x}^3 \hat{I}_2\hat{I}_4  \\
&& \quad\quad + \frac{29}{5760}\hat{x}^3  \hat{I}_3^2  
+ \frac{1}{1152} \hat{x}^3  \hat{I}_5   
+  \frac{7}{96} \hat{x}^5 \hat{I}_2^3 
+   \frac{29}{480}\hat{x}^5 \hat{I}_2\hat{I}_3 
+  \frac{1}{128} \hat{x}^5\hat{I}_4   \\
&& \quad\quad + \frac{7}{32} \hat{x}^7\hat{I}_2^2      +  \frac{29}{384} \hat{x}^7 \hat{I}_3  
+  \frac{203}{384} \hat{x}^9 \hat{I}_2 
+ \frac{105}{128} \hat{x}^{11}. 
\een

\section{The Pure Gravity: The $k=2$ Critical Point}

In this Section we consider the simplest multicritical point given by the $(3,2)$-model coupled with topological gravity.
The central charge of this model is $c=0$, so the fictitious space has dimensional $0$, so there is no effect from the fictitious dimensional and the theory
is referred as the pure gravity.
It is well-known to be related to the first Painl\'eve equation.

\subsection{The $I$-coordinates at the $k=2$ critical point}
\label{sec:k=2}

Take $t_1$ to be $t_1+1$,
$t_2$ to be $t_2-2a$.
Then the string equation in genus zero becomes
\begin{equation}
I_0= t_0 + (t_1+1) I_0 + (t_2 -2a) \frac{I_0^2}{2!}
+ t_3 \frac{I_0^3}{3!} +\cdots.
\end{equation}
Write
\be  
t_0 = x= aq^2,
\ee
then we have:
\be  \label{eqn:I0-q}
\biggl(\frac{I_0}{q}\bigg)^2= 1 +  \frac{t_1}{aq} \frac{I_0}{q} +  \frac{1}{2!}\frac{t_2}{a} \biggl(\frac{I_0}{q}\biggr)^2
+\frac{1}{3!} \frac{qt_3}{a} \biggl(\frac{I_0}{q}\biggr)^3  +\cdots.
\ee
From this one gets:
\ben
I_0 & = &q \biggl(1+ \frac{t_1}{2aq} +\frac{1}{2!} \frac{t_2}{2a} +\frac{1}{2} \biggl(\frac{t_1}{2aq}\biggr)^2 
+ \frac{1}{3!} \frac{qt_3}{2a} + \frac{t_1}{2aq} \frac{t_2}{2a} \\
& + & \frac{1}{4!}\frac{q^2t_4}{2a} + \frac{1}{2} \frac{t_1}{2aq} \frac{qt_3}{2a} + \frac{3}{8} \biggl(\frac{t_2}{2a}\biggr)^2 
+ \frac{3}{4} \biggl(\frac{t_1}{2aq}\biggr)^2\frac{t_2}{2a} 
- \frac{1}{8}\biggl( \frac{t_1}{2a}\biggr)^4 \\
& + & \frac{1}{5!}\frac{q^3t_5}{2a} + \frac{1}{6} \frac{t_1}{2aq} \frac{q^2t_4}{2a} + \frac{1}{3} \frac{t_2}{2a} \frac{qt_3}{2a} 
+ \frac{2}{3} \biggl(\frac{t_1}{2aq}\biggr)^2\frac{qt_3}{2a} + \frac{t_1}{2aq} \biggl(\frac{t_2}{2a}\biggr)^2+\cdots\biggr).
\een
This is obtained as follows.
Write
\be 
\frac{I_0}{q} = 1 + \sum_{n \geq 1} C_n(\hat{t}_1, \dots, \hat{t}_n),
\ee
where $\hat{t}_n = \frac{q^{n-2}t_n}{2a}$, and $C_n$ is a weighted 
homogeneous polynomial in $\hat{t}_1, \dots, \hat{t}_n$ where we assign $\deg \hat{t}_n = n$.
Plugging this into \eqref{eqn:I0-q},
one gets a sequence of recursion relations:
\ben
2 C_1 & = & 2 \hat{t}_1, \\
2C_2 + C_1^2 & = & 2\hat{t_1} C_1 + \hat{t}_2, \\
2C_3 + C_1C_2 & = & 2\hat{t_1} C_2 + \hat{t}_2 \cdot 2 C_1 + \frac{1}{3} \hat{t}_3, 
\een
etc.
These recursion relations uniquely all the $C_n$'s.

\begin{rmk}
Gross and Migdal \cite{Gross-Migdal} used Lagrange method to prove the following inversion formula:
If
\be 
x=y^k - \sum_i C_i y^i,
\ee
then
\be  
y(x) =x^{1/k} + \frac{1}{k}  \sum_{p=1}^\infty \frac{1}{p!} \pd_x^{p-1} \biggl[ \biggl( \sum_i C_i x^{i/k}\biggr)^p x^{1/k-1}\biggr].
\ee
Now we have
\begin{equation}
\frac{t_0}{a} = I_0^2 -  \sum_{i\geq 1} \frac{t_i}{i!a} I_0^i ,
\end{equation}
and so one gets a closed formula for $I_0$:
\ben 
I_0 & = & q + \frac{1}{2} \sum_{i=1}^\infty \frac{t_i}{i!a}   q^{i-1} 
+ \frac{1}{2} \cdot \frac{1}{2!} \cdot  \frac{1}{2q}\frac{\pd}{\pd q} \biggl[ \biggl(\sum_{i=1}^\infty \frac{t_i}{i!a} q^i \bigg)^2 q^{-1}\biggr] + \cdots.
\een
\end{rmk}

One can also assign the following degrees:
\be 
\Deg t_n  = 2- n.
\ee 
Then one has
\begin{align} 
\Deg q & = 1, & \Deg \hat{t}_n&  = 0,
\end{align}
and so $I_0$ is weighted homogeneous of degree
\be 
\Deg I_0 = 1.
\ee 
Plugging the explicit results for $I_0$ into \eqref{Def:In},
one computes $I_k$ for $k \geq 1$. 
For example,
\ben
I_1 & = & (t_1+1) + (t_2-2a) I_0 +t_3 \frac{I_0^2}{2!}
+t_4 \frac{I_0^3}{3!} + t_5\frac{I_0^4}{4!} +\cdots \\
& = & 1 -2aq \bigg(1- \frac{1}{2} \frac{t_2}{2a}+\frac{1}{2} \biggl(\frac{t_1}{2a}\biggr)^2
-\frac{1}{3} \frac{qt_3}{2a}-\frac{1}{8} \frac{q^2t_4}{2a}
-\frac{1}{2}\frac{t_1}{2aq}\frac{qt_3}{2a} \\
& - & \frac{1}{8}\biggl(\frac{t_2}{2a}\biggr)^2 
+ \frac{1}{4} \biggl(\frac{t_1}{2aq} \biggr)^2\frac{t_2}{2a}-\frac{1}{8} \biggl(\frac{t_1}{2aq}\biggr)^4 \\
& - & \frac{1}{30}\frac{q^3t_5}{2a} -\frac{1}{3}\frac{t_1}{2aq}\frac{q^2t_4}{2a}
-\frac{1}{3}\frac{t_2}{2a} \frac{qt_3}{2a}-\frac{1}{3}\biggl(\frac{t_1}{2aq}\biggr)^2\frac{qt_3}{2a} \biggr) 
+\cdots,
\een
and so $1-I_1$ is weighted homogeneous of degree
\be 
\Deg(1-I_1) = 1.
\ee
and  $\log (1-I) + \log (2aq)$ is weighted homogeneous of degree $0$.
We also have 
\ben
I_2 & = & (t_2-2a)  +t_3 I_0
+t_4 \frac{I_0^2}{2!} + t_5\frac{I_0^3}{3!} +\cdots \\
& = & -2a \biggl(1- \frac{t_2}{2a}- \frac{qt_3}{2a} - \frac{1}{2} \frac{q^2t_4}{2a}
- \frac{t_1}{2aq}\frac{qt_3}{2a} \\
& - & \frac{1}{6}\frac{q^3t_5}{2a} - \frac{t_1}{2a}\frac{q^2t_4}{2a}
-\frac{1}{2} \frac{t_2}{2a} \frac{qt_3}{2a} - \frac{1}{2} \biggl(\frac{t_1}{2aq}\biggr)^2\frac{qt_3}{2a}
 + \cdots \biggr) ,
\een

\ben
I_3 & = &  t_3  
+t_4  I_0  + t_5\frac{I_0^2}{2!} +\cdots 
=  t_3 + t_4q + \frac{aq^2t_5 + t_1t_4}{2a} +\cdots \\
& = & \frac{2a}{q} \biggl( \frac{qt_3}{2a} + \frac{q^2t_4}{2a} 
+ \frac{1}{2} \frac{q^3t_5}{2a} + \frac{t_1}{2aq} \frac{q^2t_4}{2a}+\cdots \biggr), 
\een

\ben
I_4 & = &   t_4   + t_5 I_0 +\cdots 
=  t_4 + t_5q   +\cdots = \frac{2a}{q^2} \biggl(\frac{q^2t_4}{2a} + \frac{q^3t_5}{2a}+\cdots\biggr), 
\een
etc.
Note for $n \geq 2$, $I_n $ is weighted homogeneous of degree
\be 
\Deg I_n = 2-n.
\ee

Plugging these results into  \eqref{eqn:F0}  ,
one gets:
\ben
F_0& = & \frac{I_0^3}{6}(1-I_1)^2+\sum_{p> 1} \frac{(-1)^p}{(p+2) \cdot p!} I_0^{p+2}I_p(1-I_1) 
+\frac{1}{2}\sum_{p,q>1}\frac{(-1)^{p+q}I_0^{p+q+1}}{(p+q+1)q!p!}I_pI_{q} \\
&= & \frac{4a^2q^5}{15}   
+ \frac{aq^4}{4} t_1 
+ (\frac{aq^5}{15} t_2 + \frac{q^3}{6}t_1^2) 
+ (\frac{aq^6}{72}t_3 + \frac{q^4}{8}t_1t_2 + \frac{q^2}{12a}t_1^3) \\
&  + &  (\frac{aq^7}{420}t_4 + \frac{q^5}{30}t_1t_3 
+ \frac{q^5}{4}t_2^2 + \frac{q^3}{8a} t_1^2t_2 
+ \frac{q}{32a^2} t_1^4) \\
& + & (\frac{aq^8}{2880}t_5 + \frac{q^6}{144}t_1t_4 
+ \frac{q^6}{72} t_2t_3 + \frac{q^4}{24a} t_1^2t_3 
+ \frac{q^4}{16a}t_1t_2^2 + \frac{q^2}{12a^2} t_1^3t_2
+ \frac{1}{120a^3} t_1^5) + \cdots 
\een
is weighted homogeneous of degree 
\be 
\Deg F_0 = 5.
\ee
Let us now make the following observation:
$F_0$ can be split into three parts:
\be 
F_0 = F_0^{(0)} +t_0  F^{(1)}_1 + \hat{F}_0,
\ee
where $F_0^{(0)}$ and $F^{(1)}_0$ do not depend on $t_0$,
and $\hat{F}_0$ contains terms which has a factor $t_0^\alpha$ with $\alpha\neq 0, 1$.
A term in $F_0^{(0)}$ is a constant times $t_{k_1} \cdots t_{k_n}$ with $k_1, \dots, k_n\geq 1$.
It degree is 
\be 
(2-k_1) + \cdots + (2-k_n) = 5,
\ee
and so 
\be 
(k_1-1) + \cdots + (k_n-1) = n -5. 
\ee
This has a solution only if $n \geq 5$. 
The first couple of terms of $F^{(0)}_0$ are
\be
F^{(0)}_0 = \frac{1}{120a^3} t_1^5 + \frac{1}{80a^4} t_1^4t_2 + \cdots.
\ee
A term in $F_0^{(1)}$ is a constant times $t_0t_{k_1} \cdots t_{k_n}$ with $k_1, \dots, k_n\geq 1$.
It degree is 
\be 
2+ (2-k_1) + \cdots + (2-k_n) = 5,
\ee
and so 
\be 
(k_1-1) + \cdots + (k_n-1) = n -3. 
\ee
This has a solution only if $n \geq 3$. 
The first few terms of $F^{(1)}_0$ are
\be
F^{(1)}_0 =  \frac{q^2}{12a} t_1^3 + \frac{q^2}{12a^2} t_1^3t_2 + \cdots.
\ee

Similarly, 
using \eqref{eqn:F1} one gets:
\ben
F_1 &=& -\frac{1}{24} \log(2aq)
+(\frac{1}{96a}t_2 - \frac{1}{192a^2q^2}t_1^2) 
+ \frac{q}{144a}t_3 \\
& + & (\frac{q^2}{384a}t_4 + \frac{1}{192a^2}t_1t_3
 + \frac{1}{384a^2}t_2^2 - \frac{1}{384a^3q^2}t_1^2t_2 
 + \frac{1}{1536a^4q^4} t_1^4) \\
&  + & (\frac{q^3}{1440a}t_5 
+ \frac{q}{288a^2} t_1t_4 
+ \frac{q}{192a^2} t_2t_3 + \frac{1}{1152a^3q}t_1^2t_3) +\cdots,
\een 
As in the case of $F_0$, 
we define $F_1^{(0)}$ and $F_1^{(1)}$,
and set $\hat{F}_1 = F_1 - F^{(0)}_1 - F^{(1)}_1$.
A term in $F_1^{(0)}$ is a constant times $t_{k_1} \cdots t_{k_n}$ with $k_1, \dots, k_n\geq 1$.
It degree is 
\be 
(2-k_1) + \cdots + (2-k_n) = 0,
\ee
and so 
\be 
(k_1-1) + \cdots + (k_n-1) = n. 
\ee
This has a solution only if $n \geq 1$. 
The first few  terms of $F^{(0)}_1$ are
\be
F^{(0)}_1 = \frac{t_2}{96a}   + \frac{ t_1t_3}{192a^2}  + \frac{t_2^2}{384a^2}  + \frac{t_1^2t_4}{512a^3}
+ \frac{t_1t_2t_3}{192a^3} + \frac{t_2^3}{1152a^3}  + \cdots.
\ee
A term in $F_1^{(1)}$ is a constant times $t_0t_{k_1} \cdots t_{k_n}$ with $k_1, \dots, k_n\geq 1$.
It degree is 
\be 
2+ (2-k_1) + \cdots + (2-k_n) = 0,
\ee
and so 
\be 
(k_1-1) + \cdots + (k_n-1) = n +2. 
\ee
This has a solution only if $n \geq 1$. 
The first few terms of $F^{(1)}_1$ are
\be
F^{(1)}_1 =  \frac{q^2}{384a} t_4 + \frac{q^2}{768a^2} t_1t_5 + \frac{q^2}{384a^2}t_2t_4+\frac{11q^2}{6912}t_3^2+  \cdots.
\ee

Finally, 
using \eqref{eqn:Fg} one can compute $F_g$.
For example,
\ben
F_2 &= &-\frac{7}{5760a^2q^5} + \frac{7t_2}{23040a^3q^5} + \frac{7t_1^2}{9216a^4q^7}  
+ \frac{5t_3}{27648a^3q^4} + \frac{t_4}{92160a^3q^3} \\
& + & \frac{7t_1t_3}{46080a^4q^5} + \frac{7t_2^2}{184320a^4q^5} + \frac{7t_1^2t_2}{36864a^5q^7}
 - \frac{49t_1^4}{147456a^6q^9} \\
& - & \frac{t_5}{276480a^3q^2}  + \frac{5t_1t_4}{55296a^4q^4} + \frac{5t_2t_3}{55296a^4q^4} - \frac{5t_1^2t_3}{55296a^5q^6} + \cdots.
\een
It is easy to see that for $g\geq 2$, $F_g$ is weighted homogeneous of degree 
\be 
\Deg F_g = 5(1-g).
\ee
So if we set
\be 
\Deg \lambda = \frac{5}{2},
\ee
then $F+ \frac{1}{24} \log (2aq)$ is weighted homogeneous of degree $0$.
We have a similar decomposition
\be 
F_g = F_g^{(0)} + F_g^{(1)} + \hat{F}_g.
\ee

\subsection{Free energy and correlation functions   when $t_n =0$ for $n \geq 1$}

The fact that $F_g^{(0)}$ and $F_g^{(1)}$ are nonvanishing makes it difficult 
to compute the free energy from the specific heat $u= \lambda^2 \pd_{t_0}^2F$. 
Note for $g\geq 0$,
each term of $F_g^{(0)}$ and $F_g^{(1)}$ contains some $t_k$ for $k \geq 1$,
and so
\begin{align} 
F_g^{(0)}(t_0) &:= F_g^{(0)}|_{t_k =0, k\geq 1} = 0, & F_g^{(1)}(t_0): = F_g^{(1)}|_{t_k =0, k\geq 1} = 0.
\end{align}
The first few terms of $F(t_0):=F|_{t_n = 0, n \geq 1}$ are given by:
 \be \label{eqn:F(t0)}
\lambda^2F(t_0) =\frac{4}{15}\frac{t_0^{5/2}}{a^{1/2}}-\frac{\lambda^2}{24}\ln(2a^{1/2}t_0^{1/2})
-\frac{7a^{1/2}\lambda^4}{5760t_0^{5/2}}   + \frac{245a\lambda^6}{331776t_0^5}
+ \cdots,
\ee
in particular,
\bea
F_0(t_0) & = & \frac{4}{15}\frac{t_0^{5/2}}{a^{1/2}}, \\
F_1(t_0) &= & -\frac{1}{48}\ln(t_0), \\
F_2(t_0) & = & -\frac{7a^{1/2}}{5760t_0^{5/2}}. \\
F_3(t_0)  & = &  \frac{245a}{331776t_0^5}.
\eea 
 
 When $t_n =0$ for $n \geq 1$, 
 \ben
&&   \nabla(x_1)\nabla(x_2)\nabla(x_3)F_0 (t_0) = \frac{1}{2aq(x_1-2q)^{\frac{3}{2}}(x_2-2q)^{\frac{3}{2}}(x_3-2q)^{\frac{3}{2}} }, \\
&& \nabla(x_1)\cdots\nabla(x_4)F_0(t_0) = \frac{1}{4a^2q^3\prod_{i=1}^4(x_i-2q)^{\frac{3}{2}}} 
\biggl(-1 + \sum_{i=1}^4 \frac{3q}{x_i-2q}\biggr), \\ 
&& \nabla(x)F_1(t_0) = 
- \frac{1}{48aq^2 (x-2q)^{\frac{3}{2}}} + \frac{1}{16aq (x-2q)^{\frac{5}{2}}}, \\
&&  \nabla(x_1)\nabla(x_2)F_1 = \frac{5}{32a^2q^2(x_1-2)^{\frac{3}{2}} (x_2-2q)^{\frac{7}{2}}}  \\
 && \quad \quad+ \frac{3}{32a^2q^2(x_1-2)^{\frac{5}{2}} (x_2-2q)^{\frac{5}{2}}} + \frac{5}{32a^2q^2(x_1-2)^{\frac{3}{2}} (x_2-2q)^{\frac{7}{2}}}   \\
&&\quad \quad - \frac{1}{16a^2q^3}(x_1-2q)^{\frac{3}{2}} (x_2-2q)^{\frac{5}{2}}-  \frac{1}{16a^2q^3}(x_1-2q)^{\frac{5}{2}} (x_2-2q)^{\frac{3}{2}} \\
&&\quad \quad   + \frac{1}{48a^2q^4 (x_1-2q)^{\frac{3}{2}} (x_2-2q)^{\frac{3}{2}}} , \\
&& \nabla(x)F_2  =   \frac{7}{2304a^3q^7(x-2q)^{\frac{3}{2}}} - \frac{7}{768a^3q^6(x-2q)^{\frac{5}{2}}}
+ \frac{7}{256  a^3q^5 (x-2q)^{\frac{7}{2}} } \\
&& \quad \quad - \frac{203}{3072 a^3 q^4  (x-2q)^{\frac{9}{2}}  }
+ \frac{135}{1024q^3(x-2q)^{\frac{11}{2}}}. 
\een
These results match with the results obtained in \cite{Bergere-Eynard, Iwaki-Sanchez}
by topological recursions.
In particular,
the free energy in those computations now have an interpretation
simply as the restriction of the free energy of the Witten-Kontsevich tau-function 
restricted to $t_n = \delta_{n,1}$ for $n \geq 1$.

\section{The Yang-Lee Edge Singularity: The $k=3$ Critical Point}
\label{sec:k=3} 

We consider in this Section the $(5,2)$-model coupled with topological gravity.
The $(5,2)$-model arises at the Yang-Lee edge singularity (see \cite{Gao-etal}and the reference therein.) 

\subsection{The $k=3$ critical point}

Now take $t_1$ to be $t_1+1$,
$t_3$ to be $t_3-6a$.
Then the string equation in genus zero becomes:
\begin{equation}
I_0= t_0 + (t_1+1) I_0 + t_2 \frac{I_0^2}{2!}
+ (t_3-6a) \frac{I_0^3}{3!} +\cdots.
\end{equation}
Write
\be  
t_0 = x= aq^3,
\ee
then we have:
\be 
\biggl(\frac{I_0}{q}\bigg)^3= 1 +  \frac{t_1}{aq^2} \frac{I_0}{q} 
+ \frac{1}{2!} \frac{t_2}{aq} \biggl(\frac{I_0}{q}\biggr)^2
+ \frac{1}{3!} \frac{t_3}{a} \biggl(\frac{I_0}{q}\biggr)^3
+ \frac{1}{4!} \frac{qt_3}{a} \biggl(\frac{I_0}{q}\biggr)^4  +\cdots.
\ee
From this one gets:
\ben
I_0 & = &q \biggl(1+ \frac{t_1}{3aq^2} +\frac{1}{2!} \frac{t_2}{3aq} 
+ \frac{1}{3!} \frac{t_3}{3a} + \frac{1}{2} \frac{t_1}{3aq^2} \frac{t_2}{3aq} - \frac{1}{3}\biggl (\frac{t_1}{3aq^2}\biggr)^3 \\
& + & \frac{1}{4!}\frac{qt_4}{3a} + \frac{1}{3} \frac{t_1}{3aq^2} \frac{t_3}{3a} + \frac{1}{4} \biggl(\frac{t_2}{3aq}\biggr)^2 
- \frac{1}{2} \biggl(\frac{t_1}{3aq^2}\biggr)^2\frac{t_2}{3aq} 
+ \frac{1}{3}\biggl( \frac{t_1}{3aq^2}\biggr)^4 \\
& + & \frac{1}{5!}\frac{q^2t_5}{3a} + \frac{1}{8} \frac{t_1}{3aq^2} \frac{qt_4}{3a} + \frac{1}{4} \frac{t_2}{3aq} \frac{t_3}{3a} 
 +\cdots\biggr).
\een 
Note as in the $k=2$ case, one can write
\be 
\frac{I_0}{q} = 1 + \sum_{n \geq 1} C_n(\hat{t}_1, \dots, \hat{t}_n),
\ee
where $\hat{t}_n = \frac{q^{n-3}t_n}{3a}$, and $C_n$ is a weighted 
homogeneous polynomial in $\hat{t}_1, \dots, \hat{t}_n$ where we assign $\deg \hat{t}_n = n$.
One can also assign the following degrees:
\be 
\Deg t_n  = 3- n.
\ee 
Then one has
\begin{align} 
\Deg q & = 1, & \Deg \hat{t}_n&  = 0,
\end{align}
and so $I_0$ is weighted homogeneous of degree
\be 
\Deg I_0 = 1.
\ee 
Plugging the explicit results for $I_0$ into \eqref{Def:In},
one computes $I_k$ for $k \geq 1$. 
For example,
\ben
I_1 & = & (t_1+1) + t_2 I_0 + (t_3-6a) \frac{I_0^2}{2!}
+t_4 \frac{I_0^3}{3!} + t_5\frac{I_0^4}{4!} +\cdots \\
& = & 1-t_1 -3aq^2 \bigg(1+\biggl(\frac{t_1}{3aq^2}\biggr)^2
-\frac{1}{6} \frac{t_3}{3a}+  \frac{t_1}{3aq^2} \frac{t_2}{3aq}-\frac{2}{3} \biggl( \frac{t_2}{3aq^2} \biggr)^3  \\
& - & \frac{1}{12} \frac{qt_4}{3a}
+  \frac{1}{4}\biggl(\frac{t_2}{3aq}\biggr)^2  
-  \frac{1}{40}\frac{q^2t_5}{3a} -\frac{1}{6}\frac{t_1}{3aq^2}\frac{qt_4}{3a}
+\frac{1}{6}\biggl(\frac{t_1}{3aq^2}\biggr)^2\frac{t_3}{3a} \\
& + & \frac{1}{2} \frac{t_1}{3aq^2} \biggl(\frac{t_2}{3aq}\biggr)^2 
- \biggl(\frac{t_1}{3aq^2}\biggr)^3 \frac{t_2}{3aq}  + \frac{2}{3} \biggl(\frac{t_1}{3aq^2}\biggr)^5 \biggr) 
+\cdots,
\een
and so $1-I_1$ is weighted homogeneous of degree
\be 
\Deg(1-I_1) = 2.
\ee
and  $\log (1-I) + \log (3aq^2)$ is weighted homogeneous of degree $0$.
We also have 
\ben
I_2 & = & t_2  + (t_3-6a) I_0
+t_4 \frac{I_0^2}{2!} + t_5\frac{I_0^3}{3!} +\cdots \\
& = & -6a q\biggl(1+ \frac{t_1}{3aq^2}- \frac{1}{3}\frac{t_3}{3a} + \frac{1}{2} \frac{t_1}{3aq^2}\frac{t_2}{3aq}
 - \frac{1}{3} \biggl( \frac{t_1}{3aq^2}\biggr)^3 \\
& - &  \frac{5}{24} \frac{qt_4}{3a}
- \frac{1}{6} \frac{t_1}{3aq^2}\frac{t_3}{3a} + \frac{1}{4} \biggl( \frac{t_2}{3aq}\biggr)^2 - \frac{1}{2} \biggl(\frac{t_1}{3aq^2}\biggr)^2\frac{t_2}{3aq}
+ \frac{1}{3} \biggl( \frac{t_1}{3aq^2}\biggr)^4  \\
& - & \frac{3}{40}\frac{q^2t_5}{3a} - \frac{3}{8} \frac{t_1}{3aq^2}\frac{qt_4}{3a}
 + \cdots \biggr) ,
\een

\ben
I_3 & = &  (t_3-6a)  
+t_4  I_0  + t_5\frac{I_0^2}{2!} +\cdots 
=-6a +  t_3 + t_4q + \frac{3aq^3t_5 + 2t_1t_4}{6aq} +\cdots \\
& = & -6a \biggl( 1-\frac{1}{2} \frac{t_3}{3a} -\frac{1}{2} \frac{qt_4}{3a} 
- \frac{1}{4} \frac{q^2t_5}{3a} -\frac{1}{2} \frac{t_1}{3aq^2} \frac{qt_4}{3a}+\cdots \biggr), 
\een

\ben
I_4 & = &   t_4   + t_5 I_0 +\cdots 
=  t_4 + t_5q   +\cdots = \frac{3a}{q} \biggl(\frac{qt_4}{3a} + \frac{q^2t_5}{3a}+\cdots\biggr), 
\een
etc.
Note for $n \geq 2$, $I_n $ is weighted homogeneous of degree
\be 
\Deg I_n =3-n.
\ee

Plugging these results into  \eqref{eqn:F0}  ,
one gets:
\ben
F_0 
&= & \frac{9a^2q^7}{28}   
+ \frac{3aq^5}{10} t_1 
+ (\frac{aq^6}{12} t_2 + \frac{q^3}{6}t_1^2) 
+ (\frac{aq^7}{56}t_3 + \frac{q^4}{8}t_1t_2 + \frac{q}{18a}t_1^3) \\
&  + &  (\frac{aq^8}{320}t_4 + \frac{q^5}{30}t_1t_3 
+ \frac{q^5}{40}t_2^2 + \frac{q^2}{12a} t_1^2t_2 
+ \frac{1}{108a^2q} t_1^4) \\
& + & (\frac{aq^9}{2160}t_5 + \frac{q^6}{144}t_1t_4 
+ \frac{q^6}{72} t_2t_3 + \frac{q^3}{36a} t_1^2t_3 
+ \frac{q^3}{24a}t_1t_2^2 + \frac{1}{36a^2} t_1^3t_2 ) + \cdots 
\een
is weighted homogeneous of degree 
\be 
\Deg F_0 = 7.
\ee
As in the $k=2$ case,
$F_0$ can be split into three parts:
\be 
F_0 = F_0^{(0)} +t_0  F^{(1)}_1 + \hat{F}_0,
\ee
where $F_0^{(0)}$ and $F^{(1)}_0$ do not depend on $t_0$,
and $\hat{F}_0$ contains terms which has a factor $t_0^\alpha$ with $\alpha\neq 0, 1$.
A term in $F_0^{(0)}$ is a constant times $t_{k_1} \cdots t_{k_n}$ with $k_1, \dots, k_n\geq 1$.
It degree is 
\be 
(3-k_1) + \cdots + (3-k_n) = 7,
\ee
and so 
\be 
(k_1-1) + \cdots + (k_n-1) = 2n -7. 
\ee
This has a solution only if $n \geq 4$. 
The first  term of $F^{(0)}_0$  is 
\be
F^{(0)}_0 = \frac{1}{36a^2} t_1^3t_2  + \cdots.
\ee
A term in $F_0^{(1)}$ is a constant times $t_0t_{k_1} \cdots t_{k_n}$ with $k_1, \dots, k_n\geq 1$.
It degree is 
\be 
3+ (3-k_1) + \cdots + (3-k_n) = 7,
\ee
and so 
\be 
(k_1-1) + \cdots + (k_n-1) = 2n -4. 
\ee
This has a solution only if $n \geq 2$. 
The first few terms of $F^{(1)}_0$ are
\be
F^{(1)}_0 =  \frac{q^3}{6}t_1^2 + \frac{q^3}{36a} t_1^2t_3 
+ \frac{q^3}{24a}t_1t_2^2 + \cdots.
\ee

Similarly, 
using \eqref{eqn:F1} one gets:
\ben
F_1 &=& -\frac{1}{24} \log(3aq^2)
-\frac{1}{72aq^2}t_1 - \frac{1}{432a^2q^4}t_1^2  \\
& + &  \frac{1}{432a}t_3 -\frac{1}{216} \frac{t_1t_2}{a^2q^3} + \frac{1}{486} \frac{t_1^3}{s^^3q^6}\\
& + & (\frac{q}{864a}t_4 - \frac{t_1t_3}{1296a^2} 
 - \frac{t_2^2}{864a^2q^2}  + \frac{t_1^2t_2}{648a^3q^5}  
 - \frac{11t_1^4}{23328a^4q^8}  ) \\
&  + & (\frac{q^2t_5}{2880a}  
+ \frac{t_1t_4 }{2592a^2q}  
 -  \frac{t_1^2t_3}{3888a^3q^4}  -\frac{t_1t_2^2}{2592a^3q^4}
 +  \frac{t_1^3t_2}{1944a^4q^7} \\
 & - & \frac{13t_1^5}{87480a^5q^{10}})+\cdots,
\een 
As in the case of $F_0$, 
we define $F_1^{(0)}$ and $F_1^{(1)}$,
and set $\hat{F}_1 = F_1 - F^{(0)}_1 - F^{(1)}_1$.
A term in $F_0^{(0)}$ is a constant times $t_{k_1} \cdots t_{k_n}$ with $k_1, \dots, k_n\geq 1$.
It degree is 
\be 
(3-k_1) + \cdots + (3-k_n) = 0,
\ee
and so 
\be 
(k_1-1) + \cdots + (k_n-1) = 2n. 
\ee
This has a solution only if $n \geq 1$. 
The firs term  of $F^{(0)}_1$ is
\be
F^{(0)}_1 = \frac{t_3}{432a}  + \cdots.
\ee
A term in $F_1^{(1)}$ is a constant times $t_0t_{k_1} \cdots t_{k_n}$ with $k_1, \dots, k_n\geq 1$.
It degree is 
\be 
3+ (3-k_1) + \cdots + (3-k_n) = 0,
\ee
and so 
\be 
(k_1-1) + \cdots + (k_n-1) = 2n +3. 
\ee
This has a solution only if $n \geq 1$. 
The first  term  of $F^{(1)}_1$ is
\be
F^{(1)}_1 = -  \frac{q^3t_6}{25920a}+  \cdots.
\ee
Using \eqref{eqn:Fg} one can compute $F_g$.
For example,
\ben
F_2 &= &-\frac{1}{480a^2q^7} + \frac{5t_1}{776a^3q^9} + \frac{17t_1^2}{11664a^4q^{11}}  \\
& + &  \frac{t_3}{8640a^3q^7} + \frac{7t_1t_2}{8640a^4q^{10}} - \frac{101t_1^3}{69984a^5q^{13}} \\
& + & \frac{49t_4}{933120a^3q^6} +  \frac{5t_2^2}{93312a^4q^9} - \frac{5t_1^2t_2}{15552a^5q^{12}}
 + \frac{5t_1^4}{26244a^6q^{15}}  + \cdots.
\een
It is easy to see that for $g\geq 2$, $F_g$ is weighted homogeneous of degree 
\be 
\Deg F_g = 7(1-g).
\ee
So if we set
\be 
\Deg \lambda = \frac{7}{2},
\ee
then $F+ \frac{1}{24} \log (3aq^2)$ is weighted homogeneous of degree $0$.
We have a similar decomposition
\be 
F_g = F_g^{(0)} + F_g^{(1)} + \hat{F}_g.
\ee

\subsection{Free energy  and correlation functions when $t_n =0$ for $n \geq 1$}

The fact that $F_g^{(0)}$ and $F_g^{(1)}$ are nonvanishing makes it difficult 
to compute the free energy from the specific heat $u= \lambda^2 \pd_{t_0}^2F$. 
Note for $g\geq 0$,
each term of $F_g^{(0)}$ and $F_g^{(1)}$ contains some $t_k$ for $k \geq 1$,
and so
\begin{align} 
F_g^{(0)}(t_0) &:= F_g^{(0)}|_{t_k =0, k\geq 1} = 0, & F_g^{(1)}(t_0): = F_g^{(1)}|_{t_k =0, k\geq 1} = 0.
\end{align}
The first few terms of $F(t_0):=F|_{t_n = 0, n \geq 1}$ are given by:
\bea
F_0(t_0) & = & \frac{9}{28} a^2a^2 q^7 = \frac{9}{28}\frac{t_0^{4/3}}{a^{1/3}}, \\
F_1(t_0) &= & - \frac{1}{24} \ln (3aq^2) = -\frac{1}{48}\ln(t_0) , \\
F_2(t_0) & = & - \frac{1}{480a^2q^7} = -\frac{a^{1/3}}{480t_0^{7/3}}   , \\
F_3(t_0)  & = &  \frac{247}{217728a^4q^{14}} =  \frac{247a^{2/3}}{217728t_0^{14/3}} .
\eea
 When $t_n =0$ for $n \geq 1$, 
\ben
&&I_ 0 = q, I_1 = 1 - 3aq^2, I_2 = -6a q, I_3 = -6a
\een
and so
 \ben
&&   \nabla(x_1)\nabla(x_2)\nabla(x_3)F_0 (t_0) = \frac{1}{3aq^2(x_1-2q)^{\frac{3}{2}}(x_2-2q)^{\frac{3}{2}}(x_3-2q)^{\frac{3}{2}} }, \\
&& \nabla(x_1)\cdots\nabla(x_4)F_0(t_0) = \frac{1}{9a^2q^5\prod_{i=1}^4(x_i-2q)^{\frac{3}{2}}} 
\biggl(-2 + \sum_{i=1}^4 \frac{3q}{x_i-2q}\biggr), \\ 
&& \nabla(x)F_1(t_0) = 
- \frac{1}{36 aq^3 (x-2q)^{\frac{3}{2}}} + \frac{1}{24aq^2 (x-2q)^{\frac{5}{2}}}, \\
&&  \nabla(x_1)\nabla(x_2)F_1 = \frac{5}{72a^2q^4(x_1-2)^{\frac{3}{2}} (x_2-2q)^{\frac{7}{2}}}  \\
 && \quad \quad+ \frac{1}{24a^2q^4(x_1-2)^{\frac{5}{2}} (x_2-2q)^{\frac{5}{2}}} + \frac{5}{72a^2q^4(x_1-2)^{\frac{3}{2}} (x_2-2q)^{\frac{7}{2}}}   \\
&&\quad \quad - \frac{1}{18a^2q^5}(x_1-2q)^{\frac{3}{2}} (x_2-2q)^{\frac{5}{2}}-  \frac{1}{18a^2q^5}(x_1-2q)^{\frac{5}{2}} (x_2-2q)^{\frac{3}{2}} \\
&&\quad \quad   + \frac{1}{36a^2q^6 (x_1-2q)^{\frac{3}{2}} (x_2-2q)^{\frac{3}{2}}} , \\
&& \nabla(x)F_2  =   \frac{7}{1440a^3q^{10}(x-2q)^{\frac{3}{2}}} - \frac{41}{3240a^3q^9(x-2q)^{\frac{5}{2}}}
+ \frac{139}{5184  a^3q^8 (x-2q)^{\frac{7}{2}} } \\
&& \quad \quad - \frac{203}{5184  a^3 q^7  (x-2q)^{\frac{9}{2}}  }
+ \frac{35}{1152 q^6(x-2q)^{\frac{11}{2}}}. 
\een

\section{Critical Exponents}
\label{sec:Exponents}

\subsection{Critical exponents of the free energy at critical points}

At the $k$-th critical point,
\begin{align}
t_0 &= aq^k, & t_1 &= 1, & t_k & = k! \cdot a, & t_n & = 0 \;\;\; n \neq 0,1,k.
\end{align}
At this point the renormalized coupling constants are given by:
\begin{align}
I_0 &= q, & I_1 &= 1- k aq^{k-1}, & I_2 & = - k(k-1)aq^{k-2}, &\dots, \\
I_k & = -k! \cdot a, & I_n & = 0 \;\;\; \text{otherwise},
\end{align}
and so
\be 
\tilde{I}_{j} = \frac{-aq^{k-j}}{(kaq^{k-1})^{(2j+1)/3}} \sim \frac{1}{q^{(2k+1)(j-1)/3}}.
\ee
It follows that
\ben 
\corr{\tau_2^{l_2}\cdots \tau_k^{l_k}}_g \tilde{I}_2^{l_2}\cdots \tilde{I}_k^{l_k} \sim q^{ (1-g)(2k+1)} 
\sim t_0^{(1-g) (1/k+2)}.
\een
Here we use the selection rule
\be 
\sum_{j=2}^k l_j(j-1) = 3g-3
\ee
to get:
\ben
&& \sum_{j=2}^k l_j \cdot (2k+1)(j-1)/3
= (2k+1)(g-1).
\een
So we have proved:
At the $k$-th critical point,
\be 
F_g(t_0)\sim t_0^{(1-g) (1/k+2)}.
\ee

\subsection{Critical exponents of the correlation functions at critical points}

A typical term of the $n$-point function $\nabla(x_1) \cdots \nabla(x_n)F_g$ is
\ben
&& (1-I_1)^{\frac{2n}{3}}\hat{x}_1^{2k_1+3} \cdots \hat{x}_n^{2k_n+3} \hat{I}_2^{l_2} \cdots \hat{I}_k^{l_k} \\
& = &  (1-I_1)^{\frac{2n}{3}} \prod_{j=1}^n (x_j-2I_0)^{-\frac{2k_j+3}{2}}(1-I_1)^{-\frac{2k_j+3}{3}} \cdot  \hat{I}_2^{l_2} \cdots \hat{I}_k^{l_k},
\een
such that
\be 
(2k_1+3)  + \cdots + (2k_n+3) +2 l_2(2-1) + \cdots +2 l_k(k-1) = 6g-6+5n.
\ee
Such a term can be expanded as follows:
\ben
&&  (1-I_1)^{\frac{2n}{3}} \prod_{j=1}^n (1-I_1)^{-\frac{2k_j+3}{3}} \sum_{n_j \geq 0} x_j^{-\frac{2k_j+3}{2}}  \binom{-k_j-1/2}{n_j} \frac{(2I_0)^{n_j}}{x_j^{n_j}}
\cdot  \prod_{j=2}^k \hat{I}_j^{l_j}, 
\een
its coefficient of $\prod_{j=1}^n x_j^{-k_j-n_j-\frac{3}{2}}$ has the following behavior:
\ben
&&  (1-I_1)^{\frac{2n}{3}} \prod_{j=1}^n (1-I_1)^{-\frac{2k_j+3}{3}}   I_0^{n_j} 
\cdot  \prod_{j=2}^k \hat{I}_j^{l_j} \\
& \sim & q^{\frac{2n}{3}(k-1)} \prod_{j=1}^n (q^{-(k-1) \cdot \frac{2k_j+3}{3}}\cdot q^{n_j}) \cdot \prod_{j=2}^k q^{-l_j \cdot (2k+1)(j-1)/3}.
\een
The total exponent is 
\ben
&& \frac{2n}{3}(k-1) + \sum_{j=1}^n  (-(k-1) \cdot \frac{2k_j+3}{3} + n_j )- \sum_{j=2}^k l_j \cdot (2k+1)(j-1)/3 \\
& = &  \frac{2n}{3}(k-1) + \sum_{j=1}^n  (-(k-1) \cdot \frac{2k_j+3}{3} + n_j ) \\
& +&  \frac{2k+1}{6} \sum_{j=1}^n (2k_j+3)- \frac{2k+1}{6} (6g-6+5n)  \\
& = & \sum_{j=1}^n (k_j+n_j -k)  + (1-g) (2k+1).
\een
So we have shown that
\be
\begin{split}
\frac{\pd^n F_g}{\pd t_{m_1}\cdots t_{m_n}}(t_0) 
\sim q^{\sum_{j=1}^n (m_j -k)  + (1-g) (2k+1)} \\
\sim t_0^{\sum_{j=1}^n (m_j/k -1)  + (1-g) (2+1/k)}.
\end{split}
\ee

\section{Conclusions and Prospects}
 
In  this work we have studied the phase transition and critical phenomena in two-dimensional topological gravity
within the Landau-Ginzburg-Wilson paradigm.
The Landau-Ginzburg equation is derived from the puncture equation and the KdV hierarchy
satisfied by the Witten-Kontsevich partition function \cite{Dijkgraaf-Witten}.
In genus zero the Landau-Ginzburg potential 
is of the form 
\be 
- \frac{1}{2} v^2 +\sum_{n=0}^\infty t_n \frac{v^{n+1}}{(n+1)!}. 
\ee
As shown in \cite{Zhou-1D},
after a simple renormalization procedure one gets: 
\be 
- \frac{1}{2} v^2 +\sum t_n \frac{v^{n+1}}{(n+1)!}
=I_{-1} + \sum_{=1}^\infty (I_n-\delta_{n,1}) \frac{(v-I_0)^n}{(n+1)!}, 
\ee
where $I_0$ is the solution of
\be 
I_0 = \sum_{n=0}^\infty t_n \frac{I_0^n}{n!}, 
\ee
and for $k \geq 1$,
\be \label{def:Ik}
I_k= \sum_{n \geq 0} t_{n+k} \frac{I_0^n}{n!},
\ee
and 
\be
I_{-1} = \sum_{n =0}^\infty t_n \frac{I_0^{n+1}}{(n+1)!}.
\ee
The renormalized coupling constants are used to get unified computations to compute 
the free energy functions and the correlations functions at all the critical points,
and the critical components can also be obtained in a simple unified way.

Usually one works locally at a  single critical point.
Surprisingly we can work uniformly at all critical points in this case. 
We believe it is a common feature for any
theory obtained from two-dimensional topological gravity coupled with some topological matter.
Indeed,
as shown in \cite{Dijkgraaf-Witten},
one can derive the Landau-Ginzburg equations from the topological recursion relations in genus zero for such a theory.
Starting  with the Landau-Ginzburg potential for these equations,
one can perform renormalization as in the case of  Witten-Kontsevich tau-function.
The puncture equation and the dilaton equation hold for such a theory.
If the puncture operator and the dilaton operator can be reformulated as simple operators in the renormalized coupling constants,
then one can again expect the  free energy in genus $>0$ as finite expressions. 
The most natural extension is to consider the intersection numbers on
Witten's moduli spaces of r-spin curves
and treat all the $(p,q)$-models in the same fashion.
In a forthcoming paper we will work out the details for case of $r=3$. 
Further extension to the general case and topological gravity coupled to topological sigma models 
will be investigated  in the future.

\vspace{.2in}
{\bf Acknowledgements}. This research is partly supported by NSFC (No.~12371254, No.~11890662, No.~12061131014).
The author thanks Professor Yongbin Ruan and Qile Chen for  helpful discussions on $(p,q)$-minimal models many years ago.
These discussions have kept the author curious about such models.

\bibliographystyle{plain}

\end{document}